\documentclass[a4paper,twocolumn,11pt,accepted=2024-08-16]{quantumarticle}
\pdfoutput=1

\usepackage[T1]{fontenc}
\usepackage{amsmath,amssymb} 
\usepackage{amsthm}
\usepackage{amsfonts}
\usepackage{graphicx}
\usepackage{hyperref}
\usepackage{xcolor}
\usepackage[ruled, vlined, linesnumbered]{algorithm2e}
\usepackage{tikz}

\newenvironment{algo}[1]{
  \algorithm[ht]
    \caption{#1}
    \DontPrintSemicolon
    \SetAlgoCaptionLayout{left}
    \SetAlgoHangIndent{0pt}
    \SetKwInOut{Input}{input}
    \SetKwInOut{Output}{output} 
}{
  \endalgorithm
}

\theoremstyle{definition}

\newtheorem{lemma}{Lemma}
\newtheorem{proposition}{Proposition}

\newcommand{\Z}{\mathbb{Z}}

\newcommand{\Plog}{P_{\log}}
\newcommand{\HH}{{\bf H}}

\DeclareMathOperator{\supp}{supp}
\DeclareMathOperator{\HGP}{HGP}

\begin{document}

\title{Fast erasure decoder for hypergraph product codes}

\author{Nicholas Connolly}
\affiliation{Inria, Paris, France}
\affiliation{Okinawa Institute of Science and Technology Graduate University, Okinawa, Japan}

\author{Vivien Londe}
\affiliation{Microsoft, Paris, France}

\author{Anthony Leverrier}
\affiliation{Inria, Paris, France}

\author{Nicolas Delfosse}
\affiliation{Microsoft Quantum, Redmond, Washington 98052, USA}

\maketitle

\begin{abstract}
We propose a decoder for the correction of erasures with hypergraph product codes, which form one of the most popular families of quantum LDPC codes. Our numerical simulations show that this decoder provides a close approximation of the maximum likelihood decoder that can be implemented in $O(N^2)$ bit operations where $N$ is the length of the quantum code. A probabilistic version of this decoder can be implemented in $O(N^{1.5})$ bit operations.
\end{abstract}

\maketitle

\section{Introduction}

Due to the high noise rate of quantum hardware, extensive quantum error correction is necessary to scale quantum devices to the regime of practical applications.
The surface code~\cite{dennis2002topological, fowler2012surface} is one of the most popular quantum error correcting codes for quantum computing architectures but it comes with an enormous qubit overhead because each qubit must be encoded into hundreds or thousands of physical qubits.

Quantum Low-Density Parity-Check (LDPC) codes~\cite{gallager1962low, mackay2004sparse} such as hypergraph product (HGP) codes~\cite{tillich2013quantum} promise a significant reduction of this qubit overhead~\cite{got14, fawzi2018constant}.
Numerical simulations with circuit noise show a $15\times$ reduction of the qubit count in the large-scale regime~\cite{tremblay2022constant}.
For applications to quantum fault tolerance, HGP codes must come with a {\em fast} decoder, whose role is to identify which error occurred.
In this work, we propose a fast decoder for the correction of erasures or detectable qubit loss. By replacing erased qubits with uniformly random mixed states, erasure correction can be converted into error correction. This is preferable to the standard error channel because non-erased qubits are assumed to have no errors. Our numerical simulations show that our decoder achieves a logical error rate close to the maximum likelihood (ML) decoder.

Our motivation for focusing on the decoding of erasures is twofold. First it is practically relevant and it is the dominant source of noise in some quantum platforms such as photonic systems~\cite{knill2001scheme, bartolucci2023fusion} for which a photon loss can be interpreted as an erasure.
Erasure errors are also present in platforms based on neutral atoms~\cite{wu2022erasure}, trapped ions~\cite{kang2023quantum}, superconducting qubits~\cite{kubica2023erasure}, or circuit QED~\cite{tsunoda2023error}.
Second, many of the ideas that led to the design of capacity-achieving classical LDPC codes over binary symmetric channels were first discovered by studying the correction of erasures~\cite{kudekar2013spatially, richardson2008modern}.

\section{Classical erasure decoders}

A linear code with length $n$ is defined to be the kernel $C = \ker H$ of an $r \times n$ binary matrix $H$ called the {\em parity-check matrix}.
Our goal is to protect a {\em codeword} $x \in C$ against erasures.
We assume that each bit is erased independently with probability $p$ and erased bits are flipped independently with probability $1/2$.
The set of erased positions is known and is given by an {\em erasure vector} $\varepsilon \in \Z_2^n$ such that bit $b_i$ is erased iff $\varepsilon_i = 1$.
The initial codeword $x$ is mapped onto a vector $y = x+e \in \Z_2^n$ where $e$ is the indicator vector of the flipped bits of $x$. In particular the support of $e$ satisfies $\supp(e) \subseteq \supp(\varepsilon)$.
To detect $e$, we compute the {\em syndrome} $s = H y = H e \in \Z_2^r$. A non-trivial syndrome indicates the presence of bit-flips.

The goal of the decoder is to provide an estimation $\hat e$ of $e$ given $s$ and $\varepsilon$ and it succeeds if $\hat e = e$.
This can be done by solving the linear system $H \hat e = s$ with the condition $\supp(\hat e) \subseteq \supp(\varepsilon)$ thanks to Gaussian elimination.
This {\em Gaussian decoder} runs in $O(n^3)$ bit operations which may be too slow in practice for large $n$.

\begin{algo}{Classical peeling decoder}
\Input{An erasure vector $\varepsilon \in \Z_2^N$ and a syndrome $s \in \Z_2^r$.}
\Output{Either {\bf failure} or $\hat e \in \Z_2^n$ such that $H \hat e = s$ and $\supp(\hat{e}) \subseteq \supp(\varepsilon)$.}
\label{algo:classical_peeling_decoder}

\BlankLine
	Set $\hat e = 0$.\;
	\While{there exists a dangling check}
	{
		Select a dangling check $c_i$.\;
		Let $b_j$ be the dangling bit incident to $c_i$.\;
		\If{$s_i = 1$}
		{
			Flip the $j$-th bit of $\hat e$.\;
			Flip $s_k$ for all checks $c_k$ incident with $b_j$.\;
		}
		Set $\varepsilon_j = 0$.
	}
	{\bf if} $\varepsilon \neq 0$ {\bf return Failure}, {\bf else return $\hat e$.}
\end{algo}

The classical peeling decoder~\cite{luby2001efficient}, described in Algorithm~\ref{algo:classical_peeling_decoder}, provides a fast alternative to the Gaussian decoder. Unlike a maximum likelihood decoder, for which logical errors are the only source of failures, the additional possibility of a decoder failure is a major drawback for the peeling decoder.  However, because peeling decoder failures are infrequent for sparse codes, this algorithm gives a much more efficient linear-complexity decoding algorithm for LDPC codes without a large increase in the failure rate.

To describe this decoder, it is convenient to introduce the {\em Tanner graph}, denoted $T(H)$, of the linear code $C = \ker H$. It is the bipartite graph with one vertex $c_1, \dots, c_r$ for each row of $H$ and one vertex $b_1, \dots, b_n$ for each column of $H$ such that $c_i$ and $b_j$ are connected iff $H_{i, j} = 1$. We refer to $c_i$ as a {\em check node} and $b_j$ as a {\em bit node}.
The codewords of $C$ are the bit strings such that the sum of the neighboring bits of a check node is 0 mod 2.
Given an erasure vector $\varepsilon$, a check node is said to be a {\em dangling check} if it is incident to a single erased bit. We refer to this erased bit as a {\em dangling bit}. Figure~\ref{fig:peeling_decoder} shows a simple example of this setup.
The basic idea of the peeling decoder is to use dangling checks to recover the values of dangling bits and to repeat until the erasure is fully corrected.

\begin{figure}[t]
\centering
{
\begin{tabular}{ccl}$H$&=&$\begin{bmatrix}1&1&1&0\\1&1&1&1\\0&1&1&1\end{bmatrix}$\\ \\$\tikz{\node [rectangle,draw,thick] at (0,0) {};}$&=&check node\\$\tikz{\node [circle,draw,thick] at (0,0) {};}$&=&bit node\\ \\$\tikz{\node [rectangle,draw,thick,fill=red] at (0,0) {};}$&=&dangling check\\$\tikz{\node [circle,draw,thick,fill=red] at (0,0) {};}$&=&dangling bit\end{tabular}
}
{
\begin{tikzpicture}[baseline=(current bounding box.center)]
    \node (1) [circle,draw,inner sep=3pt,line width=1] at (0,3) {};
    \node (2) [circle,draw,inner sep=3pt,line width=1,fill=red] at (0,2) {};
    \node (3) [circle,draw,inner sep=3pt,line width=1,opacity=0.5,color=gray] at (0,1) {};
    \node (4) [circle,draw,inner sep=3pt,line width=1,opacity=0.5,color=gray] at (0,0) {};
    
    \node (a) [rectangle,draw,inner sep=4.5pt,line width=1] at (1.5,2.5) {};
    \node (b) [rectangle,draw,inner sep=4.5pt,line width=1] at (1.5,1.5) {};
    \node (c) [rectangle,draw,inner sep=4.5pt,line width=1,fill=red] at (1.5,0.5) {};

    \path (1) edge[line width=1.5] (a);
    \path (1) edge[line width=1.5] (b);
    
    \path (2) edge[line width=1.5] (a);
    \path (2) edge[line width=1.5] (b);
    \path (2) edge[line width=1.5] (c);
    
    \path (3) edge[line width=1.5,opacity=0.5,color=gray] (a);
    \path (3) edge[line width=1.5,opacity=0.5,color=gray] (b);
    \path (3) edge[line width=1.5,opacity=0.5,color=gray] (c);
    
    \path (4) edge[line width=1.5,opacity=0.5,color=gray] (b);
    \path (4) edge[line width=1.5,opacity=0.5,color=gray] (c);
\end{tikzpicture}
}
{
\begin{tikzpicture}[baseline=(current bounding box.center)]
    \node (1) [circle,draw,inner sep=3pt,line width=1,opacity=0.5,color=gray] at (0,3) {};
    \node (2) [circle,draw,inner sep=3pt,line width=1] at (0,2) {};
    \node (3) [circle,draw,inner sep=3pt,line width=1] at (0,1) {};
    \node (4) [circle,draw,inner sep=3pt,line width=1,opacity=0.5,color=gray] at (0,0) {};
    
    \node (a) [rectangle,draw,inner sep=4.5pt,line width=1] at (1.5,2.5) {};
    \node (b) [rectangle,draw,inner sep=4.5pt,line width=1] at (1.5,1.5) {};
    \node (c) [rectangle,draw,inner sep=4.5pt,line width=1] at (1.5,0.5) {};

    \path (1) edge[line width=1.5,opacity=0.5,color=gray] (a);
    \path (1) edge[line width=1.5,opacity=0.5,color=gray] (b);
    
    \path (2) edge[line width=1.5] (a);
    \path (2) edge[line width=1.5] (b);
    \path (2) edge[line width=1.5] (c);
    
    \path (3) edge[line width=1.5] (a);
    \path (3) edge[line width=1.5] (b);
    \path (3) edge[line width=1.5] (c);
    
    \path (4) edge[line width=1.5,opacity=0.5,color=gray] (b);
    \path (4) edge[line width=1.5,opacity=0.5,color=gray] (c);

    \node (nw) at (-0.5,2.5) {};
    \node (ne) at (0.5,2.5) {};
    \node (se) at (0.5,0.5) {};
    \node (sw) at (-0.5,0.5) {};
    \node (ss) at (-0.75,1.5) {\rotatebox{90}{\text{stopping set}}};
    
    \path (nw) edge[line width=1,style=dashed,color=blue,opacity=0.5] (ne);
    \path (ne) edge[line width=1,style=dashed,color=blue,opacity=0.5] (se);
    \path (se) edge[line width=1,style=dashed,color=blue,opacity=0.5] (sw);
    \path (sw) edge[line width=1,style=dashed,color=blue,opacity=0.5] (nw);
    
\end{tikzpicture}
}
\caption{Two examples of an erasure-induced subgraph for a simple Tanner graph $T(H)$. Non-erased nodes are grayed-out and excluded from the subgraph.}
\label{fig:peeling_decoder}
\end{figure}

The notion of stopping set was introduced in \cite{zyablov1974decoding} to bound the failure probability of the decoder for classical LDPC codes.
A {\em stopping set} for the Tanner graph $T(H)$ is defined to be a subset of bits that contains no dangling bit, as shown in the rightmost example of Fig.~\ref{fig:peeling_decoder}.
If the erasure covers a non-empty stopping set, then Algorithm~\ref{algo:classical_peeling_decoder} returns {\bf Failure}.
Algorithm~\ref{algo:classical_peeling_decoder} gets stuck when the remaining erasure is equal to a non-empty stopping set because each check is incident to either zero or at least two erased bits.

The peeling decoder has been adapted to surface codes~\cite{delfosse2020linear} and color codes~\cite{lee2020trimming}.
In the rest of this paper, we design a fast erasure decoder inspired by the peeling decoder that applies to a broad class of quantum LDPC codes. Our design process relies on the analysis of stopping sets.
At each design iteration, we propose a new version of the decoder, identify its most common stopping sets and modify the decoder to make it capable of correcting these dominant stopping sets.

\section{Classical peeling decoder for quantum CSS codes}

A CSS code with length $N$ is defined by commuting $N$-qubit Pauli operators $S_{X, 1}, \dots, S_{X, R_X} \in \{I,X\}^{\otimes N}$ and $S_{Z, 1}, \dots, S_{Z, R_Z} \in \{I,Z\}^{\otimes N}$ called the {\em stabilizer generators}~\cite{steane1996multiple, calderbank1996good}.
We refer to the group they generate as the {\em stabilizer group} and its elements are called {\em stabilizers}.

We can correct $X$ and $Z$ errors independently with the same strategy. Therefore we focus on the correction of $X$ errors, based on the measurement of the $Z$-type stabilizer generators. This produces a {\em syndrome} $\sigma(E) \in \Z_2^{R_Z}$, whose $i^{\text{th}}$ component is $1$ iff the error $E$ anti-commutes with $S_{Z, i}$.
An error with trivial syndrome is called a {\em logical error} and a {\em non-trivial logical error} if it is not a stabilizer, up to a phase.

We assume that qubits are erased independently with probability $p$ and that an erased qubit suffers from a uniform error $I$ or $X$~\cite{grassl1997codes}. This results in an $X$-type error $E$ such that $\supp(E) \subseteq \supp(\varepsilon)$.
The decoder returns an estimate $\hat E$ of $E$ given the erasure vector $\varepsilon$ and the syndrome $s$ of $E$.
It succeeds iff $\hat E E$ is a stabilizer (up to a phase).
The {\em logical error rate} of the scheme, denoted $\Plog(p)$, is the probability that $\hat E E$ is a non-trivial logical error.

By mapping Pauli operators onto binary strings, one can cast the CSS erasure decoding problem as the decoding problem of a classical code with parity check matrix $\HH_Z$ whose rows correspond to the $Z$-type stabilizer generators.
As a result, one can directly apply the classical Gaussian decoder and the classical peeling decoder to CSS codes.
From Lemma~1 of~\cite{delfosse2020linear}, the Gaussian decoder is an optimal decoder, {\em i.e.} a Maximum Likelihood (ML) decoder, but its complexity scaling like $O(N^3)$ makes it too slow for large codes.
The peeling decoder is faster, with scaling as $O(N)$. However, the following lemma proves that, unlike its classical counterpart, stopping sets and hence decoder failures occur much more frequently for quantum codes even in the LDPC case.

\begin{lemma} [Stabilizer stopping sets] \label{lemma:stabilizer_stopping_set}
The support of an $X$-type stabilizer is a stopping set for the Tanner graph $T(\HH_Z)$.
\end{lemma}

\begin{proof}
This is because an $X$-type stabilizer commutes with $Z$-type generators, and therefore its binary representation is a codeword for the classical linear code $\ker \HH_Z$.
\end{proof}

As a consequence, the classical peeling decoder has no threshold for any family of quantum LDPC codes defined by bounded weight stabilizers. 
Indeed, if each member of the family has at least one $X$-type stabilizer with weight $w$, then the logical error rate satisfies $\Plog(p) \geq p^w$, which is a constant bounded away from zero when $N \to \infty$. 
This is in sharp contrast with the classical case for which the probability to encounter a stopping set provably vanishes for carefully designed families of LDPC codes~\cite{richardson2001efficient}.

\section{Pruned peeling decoder}

Since the peeling decoder gets stuck into stopping sets induced by the $X$-type generators, the idea is to look for such a generator $S$ supported entirely within the erasure and to remove an arbitrary qubit of the support of $S$ from the erasure.
We can remove this qubit from the erasure because either the error $E$ or its equivalent error $ES$ (also supported inside $\varepsilon$) acts trivially on this qubit. 

\begin{algo}{Pruned peeling decoder}
\Input{An erasure vector $\varepsilon \in \Z_2^N$, a syndrome $s \in \Z_2^{R_Z}$, and an integer $M$.}
\Output{Either {\bf Failure} or an $X$-type error $\hat E \in \{I, X\}^N$ such that $\sigma(\hat E) = s$ and $\supp(\hat E) \subseteq \supp(\varepsilon)$.}
\label{algo:pruned_peeling_decoder}

\BlankLine
	Set $\hat E = I$.\;
	\While{there exists a dangling generator}
	{
		Select a dangling generator $S_{Z, i}$.\;
		Let $j$ be the index of the dangling qubit incident to $S_{Z, i}$.\;
		\If{$s_i = 1$}
		{
			Replace $\hat E$ by $\hat E X_j$ and $s$ by $s + \sigma(X_j)$.\;
		}
		Set $\varepsilon_j = 0$.\;
	}
        \If{there exists a product $S$ of up to $M$ stabilizer generators among $S_{X,1}, \dots, S_{X, R_X}$ such that $\supp(S) \subseteq \supp(\varepsilon)$}
        {
                Select a qubit in $\supp(S)$ with index $j$ and set $\varepsilon_j = 0$.\;
                    Go back to step 2.
        }
        {\bf if} $\varepsilon \neq 0$ {\bf return Failure}, {\bf else return $\hat E$.}
\end{algo}

This leads to the pruned peeling decoder described in Algorithm~\ref{algo:pruned_peeling_decoder}.
To make it easier to follow, we use the terms {\em dangling generator} and {\em dangling qubit} in place of dangling check and dangling bit.
A dangling generator is a $Z$ generator in the context of correcting $X$ errors.
In order to keep the complexity of the peeling decoder linear, we look for an $X$-type stabilizer which is a product of up to up $M$ stabilizer generators where $M$ is a small constant.
For low erasure rate, we expect the erased stabilizers to have small weight and therefore a small value of $M$ should be sufficient.

\begin{figure}[t]
\centering
\includegraphics[scale=.40]{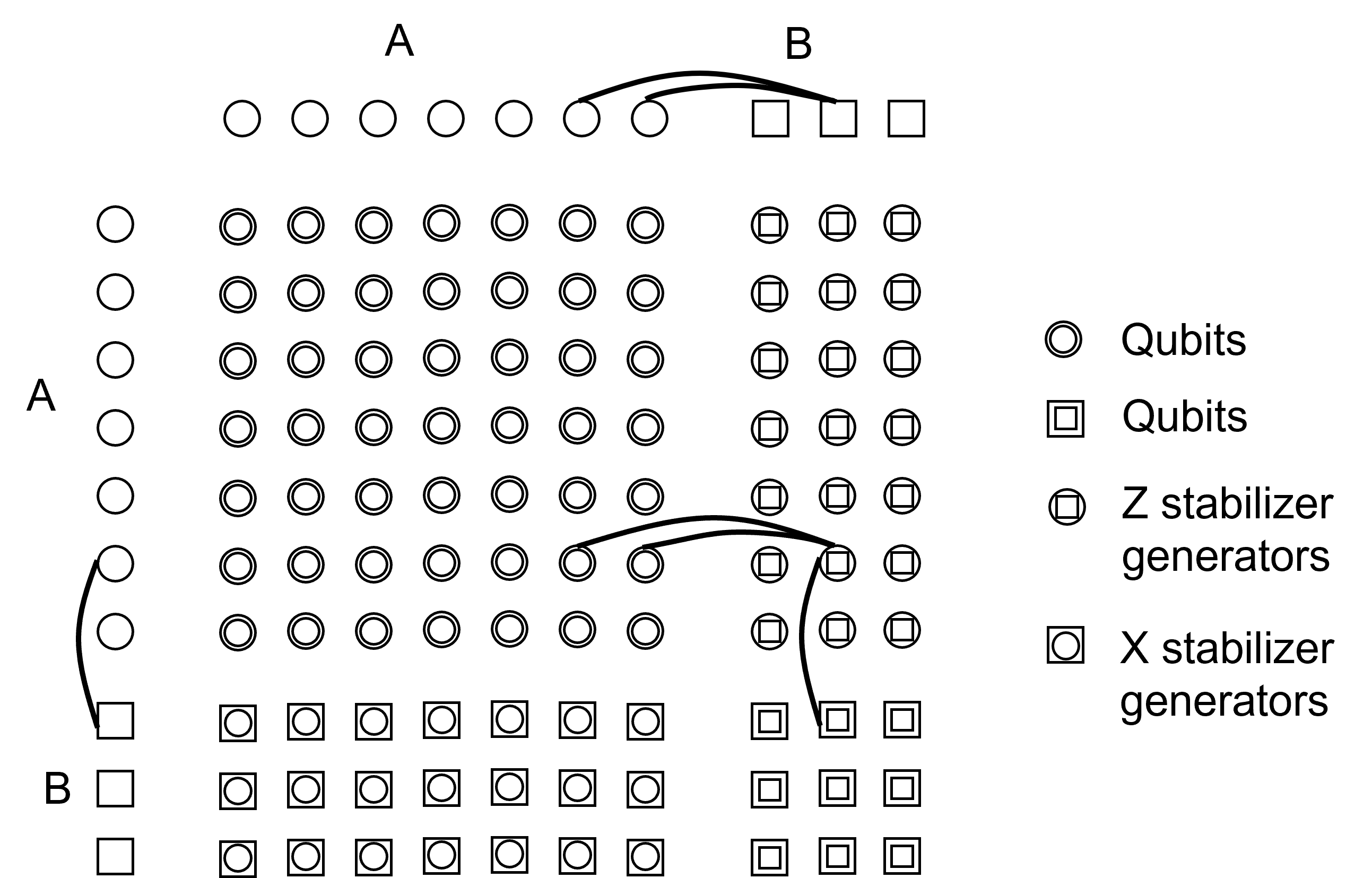}
\caption{The HGP code derived from a linear code with 7 bits and 3 checks. The support of the $Z$ stabilizer generator with index $(b, a') \in B \times A$ is given by the neighbors of $(b, a')$ in the Cartesian product of the graph $T(H)$ with itself. In the product notation, we follow the $x\times y$ convention, where the first coordinate denotes the horizontal code and the second coordinate denotes the vertical code.}
\label{fig:hgp}
\end{figure}

Although the pruned peeling algorithm can be used with any CSS code, we are particularly interested in its application to HGP codes. Here, let us recall the hypergraph product construction from~\cite{tillich2013quantum}.
For a classical code with parity check matrix $H$, recall that the Tanner graph $T(H)=(A\cup B,E_H)$ of this code is the bipartite graph whose nodes consist of disjoint sets $A$ corresponding to bits (the columns of $H$) and $B$ corresponding to checks (the rows of $H$). Edges between these nodes are defined by the corresponding nonzero cells in $H$, and are denoted by the set $E_H$.
The hypergraph product obtained from this classical code, denoted $\HGP(H)$, is defined from the cartesian product of $T(H)$ with itself (see Fig.~\ref{fig:hgp}). It is a CSS code with 4-partite Tanner graph $\HGP(H)=((A\times A)\cup (B\times B)\cup (A\times B)\cup (B\times A),E)$. Qubits in this graph correspond to the nodes in $(A\times A)\cup (B\times B)$ and come in two types resulting from the graph's product structure: "bit-bit" qubits labeled by the pairs $(a,a')\in A\times A$ and "check-check" qubits labeled by the pairs $(b,b')\in B\times B$.

The remaining nodes in $(A\times B)\cup (B\times A)$ are identified with $X$- and $Z$-type stabilizer generators, respectively. For each $(a,b')\in A\times B$, we define a stabilizer generator acting as $X$ on the qubits of the form $(b,b')\in B\times B$ such that $\{a,b\}\in E_H$ and acting as $X$ on the qubits of the form $(a,a')\in A\times A$ such that $\{a',b'\}\in E_H$. Stated more explicitly, given fixed $(a,b')\in A\times B$, any $b\in B$ adjacent to $a$ in $T(H)$ lifts to a check-check qubit $(b,b')$ adjacent to $(a,b')$ in $\HGP(H)$, and any $a'\in A$ adjacent to $b'$ in $T(H)$ lifts to a bit-bit qubit $(a,a')$ adjacent to $(a,b')$ in $\HGP(H)$; this node $(a,b')$ acts as an $X$ stabilizer generator on these two types of adjacent qubits in $\HGP(H)$, as shown in Fig.~\ref{fig:hgp}. Similarly, for each $(b,a')\in B\times A$, we define an analogous stabilizer generator acting as $Z$ on the qubits of the form $(a,a')$ such that $\{a,b\}\in E_H$ and acting as $Z$ on the qubits of the form $(b,b')$ such that $\{a',b'\}\in E_H$. If the input classical Tannar graph $T(H)$ is sparce, then $\HGP(H)$ is LDPC.

\begin{figure}[t]
\centering
\begin{tikzpicture}[baseline=(current bounding box.center)]
\node (nw) [inner sep=0] at (-4/2,4/2) {};
\node (n) [inner sep=0] at (0,4/2) {};
\node (ne) [inner sep=0] at (4/2,4/2) {};
\node (w) [inner sep=0] at (-4/2,0) {};
\node (c) [inner sep=0] at (0,0) {};
\node (e) [inner sep=0] at (4/2,0) {};
\node (sw) [inner sep=0] at (-4/2,-4/2) {};
\node (s) [inner sep=0] at (0,-4/2) {};
\node (se) [inner sep=0] at (4/2,-4/2) {};

\path (nw.center) edge[line width=1,opacity=1,dashed] (ne.center);
\path (sw.center) edge[line width=1,opacity=1,dashed] (se.center);
\path (nw.center) edge[line width=1,opacity=1,dashed] (sw.center);
\path (ne.center) edge[line width=1,opacity=1,dashed] (se.center);

\node (-33) [circle,draw,inner sep=3pt,line width=1,opacity=0.5,color=gray] at (-3/2,3/2) {};
\node (-33s) [circle,draw,inner sep=2pt,line width=1,opacity=0.5,color=gray] at (-3/2,3/2) {};
\node (-23) [circle,draw,inner sep=3pt,line width=1,opacity=0.5,color=gray] at (-2/2,3/2) {};
\node (-23s) [circle,draw,inner sep=2pt,line width=1,opacity=0.5,color=gray] at (-2/2,3/2) {};
\node (-13) [circle,draw,inner sep=3pt,line width=1,opacity=0.5,color=gray] at (-1/2,3/2) {};
\node (-13s) [circle,draw,inner sep=2pt,line width=1,opacity=0.5,color=gray] at (-1/2,3/2) {};
\node (-32) [circle,draw,inner sep=3pt,line width=1,opacity=0.5,color=gray] at (-3/2,2/2) {};
\node (-32s) [circle,draw,inner sep=2pt,line width=1,opacity=0.5,color=gray] at (-3/2,2/2) {};
\node (-22) [circle,draw,inner sep=3pt,line width=1,opacity=0.5,color=gray] at (-2/2,2/2) {};
\node (-22s) [circle,draw,inner sep=2pt,line width=1,opacity=0.5,color=gray] at (-2/2,2/2) {};
\node (-12) [circle,draw,inner sep=3pt,line width=1,opacity=1] at (-1/2,2/2) {};
\node (-12s) [circle,draw,inner sep=2pt,line width=1,opacity=1] at (-1/2,2/2) {};
\node (-31) [circle,draw,inner sep=3pt,line width=1,opacity=0.5,color=gray] at (-3/2,1/2) {};
\node (-31s) [circle,draw,inner sep=2pt,line width=1,opacity=0.5,color=gray] at (-3/2,1/2) {};
\node (-21) [circle,draw,inner sep=3pt,line width=1,opacity=0.5,color=gray] at (-2/2,1/2) {};
\node (-21s) [circle,draw,inner sep=2pt,line width=1,opacity=0.5,color=gray] at (-2/2,1/2) {};
\node (-11) [circle,draw,inner sep=3pt,line width=1,opacity=1] at (-1/2,1/2) {};
\node (-11s) [circle,draw,inner sep=2pt,line width=1,opacity=1] at (-1/2,1/2) {};

\node (13) [circle,draw,inner sep=3pt,line width=1,opacity=0.5,color=gray] at (1/2,3/2) {};
\node (13s) [rectangle,draw,inner sep=2pt,line width=1,opacity=0.5,color=gray] at (1/2,3/2) {};
\node (23) [circle,draw,inner sep=3pt,line width=1,opacity=0.5,color=gray] at (2/2,3/2) {};
\node (23s) [rectangle,draw,inner sep=2pt,line width=1,opacity=0.5,color=gray] at (2/2,3/2) {};
\node (33) [circle,draw,inner sep=3pt,line width=1,opacity=0.5,color=gray] at (3/2,3/2) {};
\node (33s) [rectangle,draw,inner sep=2pt,line width=1,opacity=0.5,color=gray] at (3/2,3/2) {};
\node (12) [circle,draw,inner sep=3pt,line width=1,opacity=1] at (1/2,2/2) {};
\node (12s) [rectangle,draw,inner sep=2pt,line width=1,opacity=1] at (1/2,2/2) {};
\node (22) [circle,draw,inner sep=3pt,line width=1,opacity=1] at (2/2,2/2) {};
\node (22s) [rectangle,draw,inner sep=2pt,line width=1,opacity=1] at (2/2,2/2) {};
\node (32) [circle,draw,inner sep=3pt,line width=1,opacity=0.5,color=gray] at (3/2,2/2) {};
\node (32s) [rectangle,draw,inner sep=2pt,line width=1,opacity=0.5,color=gray] at (3/2,2/2) {};
\node (11) [circle,draw,inner sep=3pt,line width=1,opacity=1] at (1/2,1/2) {};
\node (11s) [rectangle,draw,inner sep=2pt,line width=1,opacity=1] at (1/2,1/2) {};
\node (21) [circle,draw,inner sep=3pt,line width=1,opacity=1] at (2/2,1/2) {};
\node (21s) [rectangle,draw,inner sep=2pt,line width=1,opacity=1] at (2/2,1/2) {};
\node (31) [circle,draw,inner sep=3pt,line width=1,opacity=0.5,color=gray] at (3/2,1/2) {};
\node (31s) [rectangle,draw,inner sep=2pt,line width=1,opacity=0.5,color=gray] at (3/2,1/2) {};

\node (-3-1) [rectangle,draw,inner sep=4pt,line width=1,opacity=0.5,color=gray] at (-3/2,-1/2) {};
\node (-3-1s) [circle,draw,inner sep=2pt,line width=1,opacity=0.5,color=gray] at (-3/2,-1/2) {};
\node (-2-1) [rectangle,draw,inner sep=4pt,line width=1,opacity=0.5,color=gray] at (-2/2,-1/2) {};
\node (-2-1s) [circle,draw,inner sep=2pt,line width=1,opacity=0.5,color=gray] at (-2/2,-1/2) {};
\node (-1-1) [rectangle,draw,inner sep=4pt,line width=1,opacity=0.5,color=gray] at (-1/2,-1/2) {};
\node (-1-1s) [circle,draw,inner sep=2pt,line width=1,opacity=0.5,color=gray] at (-1/2,-1/2) {};
\node (-3-2) [rectangle,draw,inner sep=4pt,line width=1,opacity=0.5,color=gray] at (-3/2,-2/2) {};
\node (-3-2s) [circle,draw,inner sep=2pt,line width=1,opacity=0.5,color=gray] at (-3/2,-2/2) {};
\node (-2-2) [rectangle,draw,inner sep=4pt,line width=1,opacity=0.5,color=gray] at (-2/2,-2/2) {};
\node (-2-2s) [circle,draw,inner sep=2pt,line width=1,opacity=0.5,color=gray] at (-2/2,-2/2) {};
\node (-1-2) [rectangle,draw,inner sep=4pt,line width=1,opacity=0.5,color=gray,fill=orange] at (-1/2,-2/2) {};
\node (-1-2s) [circle,draw,inner sep=2pt,line width=1,opacity=0.5,color=gray] at (-1/2,-2/2) {};
\node (-3-3) [rectangle,draw,inner sep=4pt,line width=1,opacity=0.5,color=gray] at (-3/2,-3/2) {};
\node (-3-3s) [circle,draw,inner sep=2pt,line width=1,opacity=0.5,color=gray] at (-3/2,-3/2) {};
\node (-2-3) [rectangle,draw,inner sep=4pt,line width=1,opacity=0.5,color=gray] at (-2/2,-3/2) {};
\node (-2-3s) [circle,draw,inner sep=2pt,line width=1,opacity=0.5,color=gray] at (-2/2,-3/2) {};
\node (-1-3) [rectangle,draw,inner sep=4pt,line width=1,opacity=0.5,color=gray] at (-1/2,-3/2) {};
\node (-1-3s) [circle,draw,inner sep=2pt,line width=1,opacity=0.5,color=gray] at (-1/2,-3/2) {};

\node (1-1) [rectangle,draw,inner sep=4pt,line width=1,opacity=0.5,color=gray] at (1/2,-1/2) {};
\node (1-1s) [rectangle,draw,inner sep=2pt,line width=1,opacity=0.5,color=gray] at (1/2,-1/2) {};
\node (2-1) [rectangle,draw,inner sep=4pt,line width=1,opacity=0.5,color=gray] at (2/2,-1/2) {};
\node (2-1s) [rectangle,draw,inner sep=2pt,line width=1,opacity=0.5,color=gray] at (2/2,-1/2) {};
\node (3-1) [rectangle,draw,inner sep=4pt,line width=1,opacity=0.5,color=gray] at (3/2,-1/2) {};
\node (3-1s) [rectangle,draw,inner sep=2pt,line width=1,opacity=0.5,color=gray] at (3/2,-1/2) {};
\node (1-2) [rectangle,draw,inner sep=4pt,line width=1,opacity=1] at (1/2,-2/2) {};
\node (1-2s) [rectangle,draw,inner sep=2pt,line width=1,opacity=1] at (1/2,-2/2) {};
\node (2-2) [rectangle,draw,inner sep=4pt,line width=1,opacity=1] at (2/2,-2/2) {};
\node (2-2s) [rectangle,draw,inner sep=2pt,line width=1,opacity=1] at (2/2,-2/2) {};
\node (3-2) [rectangle,draw,inner sep=4pt,line width=1,opacity=0.5,color=gray] at (3/2,-2/2) {};
\node (3-2s) [rectangle,draw,inner sep=2pt,line width=1,opacity=0.5,color=gray] at (3/2,-2/2) {};
\node (1-3) [rectangle,draw,inner sep=4pt,line width=1,opacity=0.5,color=gray] at (1/2,-3/2) {};
\node (1-3s) [rectangle,draw,inner sep=2pt,line width=1,opacity=0.5,color=gray] at (1/2,-3/2) {};
\node (2-3) [rectangle,draw,inner sep=4pt,line width=1,opacity=0.5,color=gray] at (2/2,-3/2) {};
\node (2-3s) [rectangle,draw,inner sep=2pt,line width=1,opacity=0.5,color=gray] at (2/2,-3/2) {};
\node (3-3) [rectangle,draw,inner sep=4pt,line width=1,opacity=0.5,color=gray] at (3/2,-3/2) {};
\node (3-3s) [rectangle,draw,inner sep=2pt,line width=1,opacity=0.5,color=gray] at (3/2,-3/2) {};

\node (-35) [circle,draw,inner sep=3pt,line width=1,opacity=1] at (-3/2,5/2) {};
\node (-25) [circle,draw,inner sep=3pt,line width=1,opacity=1] at (-2/2,5/2) {};
\node (-15) [circle,draw,inner sep=3pt,line width=1,opacity=1] at (-1/2,5/2) {};
\node (15) [rectangle,draw,inner sep=4pt,line width=1,opacity=1] at (1/2,5/2) {};
\node (25) [rectangle,draw,inner sep=4pt,line width=1,opacity=1] at (2/2,5/2) {};
\node (35) [rectangle,draw,inner sep=4pt,line width=1,opacity=1] at (3/2,5/2) {};

\path (-15) edge[line width=1,bend left=40,looseness=0.75] (15);
\path (-15) edge[line width=1,bend left=40,looseness=0.75] (25);
\path (-25) edge[line width=1,bend left=40,looseness=0.75] (25);
\path (-25) edge[line width=1,bend left=40,looseness=0.75] (35);
\path (-35) edge[line width=1,bend left=40,looseness=0.75] (15);
\path (-35) edge[line width=1,bend left=40,looseness=0.75] (35);

\node (-53) [circle,draw,inner sep=3pt,line width=1,opacity=1] at (-5/2,3/2) {};
\node (-52) [circle,draw,inner sep=3pt,line width=1,opacity=1] at (-5/2,2/2) {};
\node (-51) [circle,draw,inner sep=3pt,line width=1,opacity=1] at (-5/2,1/2) {};
\node (-5-1) [rectangle,draw,inner sep=4pt,line width=1,opacity=1] at (-5/2,-1/2) {};
\node (-5-2) [rectangle,draw,inner sep=4pt,line width=1,opacity=1] at (-5/2,-2/2) {};
\node (-5-3) [rectangle,draw,inner sep=4pt,line width=1,opacity=1] at (-5/2,-3/2) {};

\path (-51) edge[line width=1,bend left=-40,looseness=0.75] (-5-1);
\path (-51) edge[line width=1,bend left=-40,looseness=0.75] (-5-2);
\path (-52) edge[line width=1,bend left=-40,looseness=0.75] (-5-2);
\path (-52) edge[line width=1,bend left=-40,looseness=0.75] (-5-3);
\path (-53) edge[line width=1,bend left=-40,looseness=0.75] (-5-1);
\path (-53) edge[line width=1,bend left=-40,looseness=0.75] (-5-3);

\path (-12) edge[line width=1,bend left=-40,looseness=0.75,color=orange,style=dashed,opacity=0.75] (-1-2);
\path (-11) edge[line width=1,bend left=-40,looseness=0.75,color=orange,style=dashed,opacity=0.75] (-1-2);
\path (-1-2) edge[line width=1,bend left=-40,looseness=0.75,color=orange,style=dashed,opacity=0.75] (1-2);
\path (-1-2) edge[line width=1,bend left=-40,looseness=0.75,color=orange,style=dashed,opacity=0.75] (2-2);

\path (-12) edge[line width=1,bend left=40,looseness=0.75] (12);
\path (-12) edge[line width=1,bend left=40,looseness=0.75] (22);
\path (-11) edge[line width=1,bend left=-40,looseness=0.75] (11);
\path (-11) edge[line width=1,bend left=-40,looseness=0.75] (21);
\path (1-2) edge[line width=1,bend left=40,looseness=0.75] (12);
\path (1-2) edge[line width=1,bend left=40,looseness=0.75] (11);
\path (2-2) edge[line width=1,bend left=-40,looseness=0.75] (22);
\path (2-2) edge[line width=1,bend left=-40,looseness=0.75] (21);

\node (b1nw) at (-1.5/2,2.5/2) {};
\node (b1ne) at (-0.5/2,2.5/2) {};
\node (b1se) at (-0.5/2,0.5/2) {};
\node (b1sw) at (-1.5/2,0.5/2) {};

\path (b1nw.center) edge [line width=1,dotted,color=blue] (b1ne.center);
\path (b1ne.center) edge [line width=1,dotted,color=blue] (b1se.center);
\path (b1se.center) edge [line width=1,dotted,color=blue] (b1sw.center);
\path (b1sw.center) edge [line width=1,dotted,color=blue] (b1nw.center);

\node (b2nw) at (0.5/2,-1.5/2) {};
\node (b2ne) at (2.5/2,-1.5/2) {};
\node (b2se) at (2.5/2,-2.5/2) {};
\node (b2sw) at (0.5/2,-2.5/2) {};

\path (b2nw.center) edge [line width=1,dotted,color=blue] (b2ne.center);
\path (b2ne.center) edge [line width=1,dotted,color=blue] (b2se.center);
\path (b2se.center) edge [line width=1,dotted,color=blue] (b2sw.center);
\path (b2sw.center) edge [line width=1,dotted,color=blue] (b2nw.center);
\end{tikzpicture}
{
\begin{tabular}{ccl}$\tikz{\node (a) [rectangle,draw,thick,inner sep=4pt,fill=orange] at (0,0) {};\node (b) [circle,draw,thick,inner sep=2pt] at (0,0) {};}$&=&chosen X-stab.\\$\tikz{\node [rectangle,draw,thick,inner sep=5pt,color=blue,dotted] at (0,0) {};}$&=&stopping set\end{tabular}
}
\caption{Example of a stabilizer stopping set for the distance 3 surface code obtained from a hypergraph product of two 3-bit repetition codes. The qubits in the support of an $X$-stabilizer are a stopping set for $T(H_Z)$. Highlighted nodes show the erasure subgraph corresponding to this stopping set.}
\label{fig:stabilizer_stopping_set}
\end{figure}

\begin{figure}
\centering
\includegraphics[scale=.35]{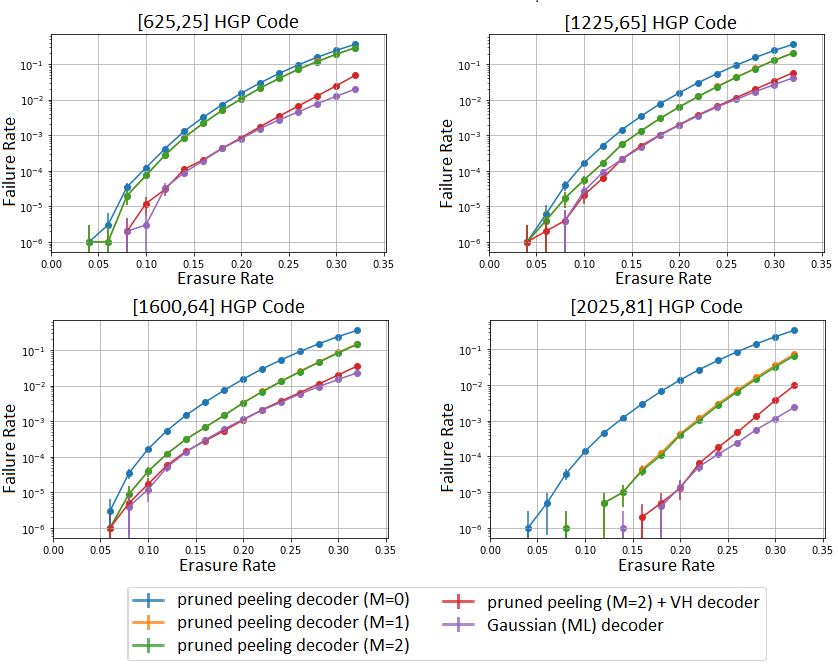}
\caption{Performance of the pruned peeling and VH decoders using four HGP codes and compared with the ML decoder ($10^6$ simulations per data point). Plots show the failure rates of the decoders for recovering an X-type Pauli error supported on the erasure vector, up to multiplication by a stabilizer.
In these plots, we discard the data points corresponding to fewer than 20 decoding failures to show only meaningful data. All the data points with at least 20 failures are shown. Note that the curves for $M=1$ and $M=2$ virtually overlap in all cases.
}
\label{fig:plots}
\end{figure}

Fig.~\ref{fig:stabilizer_stopping_set} shows an example of a stabilizer stopping set for a very simple HGP code. The pruned peeling decoder is designed to break out of stopping sets of this form. Fig.~\ref{fig:plots} shows the performance of HGP codes equipped with the pruned peeling decoder with $M = 0,1,2$. The input Tanner graphs for these codes are generated using the standard progressive edge growth algorithm which is commonly used to produce good classical or quantum LDPC codes~\cite{XEA01}. We use the implementation~\cite{XEA05,github2} of the progressive edge growth algorithm. The pruning strategy only slightly improves over the classical peeling decoder and increasing $M$ beyond $M = 1$ does not significantly affect the performance. To understand why the ML decoder severely outperforms the pruned peeling decoder, we analyze its most common stopping sets with HGP codes.

\section{Stopping sets of the pruned peeling decoder}

By studying the failure configurations of the pruned peeling decoder, we observe that the gap between the pruned peeling decoder and the ML decoder is due to the following stopping sets of HGP codes.

\begin{lemma} [Horizontal and vertical stopping sets] \label{lemma:pruned_decoder_stopping_sets}
If $R_B$ is a stopping set for a Tanner graph $T(H)$, then for all $b \in B$ the set $\{b\} \times R_B$ is a stopping set for the Tanner graph $T(\HH_Z)$ of the HGP code $\HGP(H)$.
If $R_A$ is a stopping set for a Tanner graph $T(H^T)$, then for all $a' \in A$ the set $R_A \times \{a'\}$ is a stopping set for the Tanner graph $T(\HH_Z)$ of the HGP code $\HGP(H)$.
\end{lemma}

\begin{proof}
Consider a stopping set $R_B$ for $T(H)$. Any $Z$-type stabilizer generator acting on $\{b\} \times R_B$ must be indexed by $(b, a')$ for some $a'$. Moreover, the restriction of these stabilizers to $\{b\} \times R_B$ are checks for the linear code $\ker H$. Therefore $\{b\} \times R_B$ is a stopping set for $T(\HH_Z)$.
The second case is similar.
\end{proof}

We refer to the stopping sets $\{b\} \times R_B$ as {\em vertical stopping sets} and $R_A \times \{a'\}$ are {\em horizontal stopping sets}.
Numerically, we observe that these stopping sets are responsible for the vast majority of the failures of the pruned peeling decoder.
This is because the quantum Tanner graph $T(\HH_Z)$ contains on the order of $\sqrt{N}$ copies of the type $\{b\} \times R_B$ for each stopping set $R_B$ of $T(H)$ and $\sqrt{N}$ copies of each stopping set of $T(H^T)$.
Our idea is to use the Gaussian decoders of the classical codes $\ker H$ and $\ker H^T$ to correct these stopping sets.

\section{VH decoder}

The Vertical-Horizontal (VH) decoder is based on the decomposition of the erasure into vertical subsets of the form $\{b\} \times \varepsilon_b$ with $b \in B$ and $\varepsilon_b \subseteq B$, and horizontal subsets of the form $\varepsilon_{a'} \times \{a'\}$ with $a' \in A$ and $\varepsilon_{a'} \subseteq A$, that will be decoded using the Gaussian decoder.

Let $T_v$ (resp.~$T_h$) be the subgraph of $T(\HH_Z)$ induced by the vertices of $B \times (A \cup B)$ (resp.~$(A \cup B) \times A$). The graph $T_v$ is made with the vertical edges of $T(\HH_Z)$ and $T_h$ is made with its horizontal edges.
Given an erasure vector $\varepsilon$, denote by $V(\varepsilon)$ the set of vertices of $T(\HH_Z)$ that are either erased qubits or check nodes incident to an erased qubit.
A {\em vertical cluster} (resp.~{\em horizontal cluster}) is a subset of $V(\varepsilon)$ that is a connected component for the graph $T_v$ (resp.~$T_h$).

The {\em VH graph} of $\varepsilon$ is defined to be the graph whose vertices are the clusters and two clusters are connected iff their intersection is non-empty.

The following proposition provides some insights on the structure of the VH graph.
\begin{proposition} \label{prop:VH_graph_structure}
The VH graph is a bipartite graph where each edge connects a vertical cluster with a horizontal cluster.
There is a one-to-one correspondence between (i) the check nodes of $T(\HH_Z)$ that belong to one vertical cluster and one horizontal cluster and (ii) the edges of the VH graph.
\end{proposition}

\begin{proof}
Because the graph $T_v$ contains only vertical edges, any vertical cluster must be a subset of $\{b_1\} \times (A \cup B)$ for some $b_1 \in B$. Similarly, any horizontal cluster is a subset of $(A \cup B) \times \{a_1'\}$ for some $a_1' \in A$.
As a result, two clusters with the same orientation (horizontal or vertical) cannot intersect and the only possible intersection between a cluster included in $\{b_1\} \times (A \cup B)$  and a cluster included in $(A \cup B) \times \{a_1'\}$ is the check node $(b_1, a_1')$. The bijection between check nodes and edges of the VH graph follows.
\end{proof}

\begin{figure*}[]
\centering
{
\begin{tikzpicture}[baseline=(current bounding box.center)]
\node (nw) [inner sep=0] at (-4/2,6/2) {};
\node (n) [inner sep=0] at (0,6/2) {};
\node (ne) [inner sep=0] at (3/2,6/2) {};
\node (w) [inner sep=0] at (-4/2,0) {};
\node (c) [inner sep=0] at (0,0) {};
\node (e) [inner sep=0] at (3/2,0) {};
\node (sw) [inner sep=0] at (-4/2,-4/2) {};
\node (s) [inner sep=0] at (0,-4/2) {};
\node (se) [inner sep=0] at (3/2,-4/2) {};

\path (nw.center) edge[line width=1,opacity=1] (ne.center);
\path (w.center) edge[line width=1,dashed,opacity=1] (e.center);
\path (sw.center) edge[line width=1,opacity=1] (se.center);
\path (nw.center) edge[line width=1,opacity=1] (sw.center);
\path (n.center) edge[line width=1,dashed,opacity=1] (s.center);
\path (ne.center) edge[line width=1,opacity=1] (se.center);

\node (pnw) at (0.5/2,3.5/2) {};
\node (pne) at (1.5/2,3.5/2) {};
\node (pse) at (1.5/2,-2.5/2) {};
\node (psw) at (0.5/2,-2.5/2) {};
\path (pnw.center) edge[line width=1,dotted,color=violet] (pne.center);
\path (pne.center) edge[line width=1,dotted,color=violet] (pse.center);
\path (pse.center) edge[line width=1,dotted,color=violet] (psw.center);
\path (psw.center) edge[line width=1,dotted,color=violet] (pnw.center);

\node (-35) [circle,draw,inner sep=3pt,line width=1,opacity=1] at (-3/2,5/2) {};
\node (-35s) [circle,draw,inner sep=2pt,line width=1,opacity=1] at (-3/2,5/2) {};
\node (-25) [circle,draw,inner sep=3pt,line width=1,opacity=1] at (-2/2,5/2) {};
\node (-25s) [circle,draw,inner sep=2pt,line width=1,opacity=1] at (-2/2,5/2) {};
\node (-15) [circle,draw,inner sep=3pt,line width=1,opacity=0.5,color=gray] at (-1/2,5/2) {};
\node (-15s) [circle,draw,inner sep=2pt,line width=1,opacity=0.5,color=gray] at (-1/2,5/2) {};
\node (-34) [circle,draw,inner sep=3pt,line width=1,opacity=0.5,color=gray] at (-3/2,4/2) {};
\node (-34s) [circle,draw,inner sep=2pt,line width=1,opacity=0.5,color=gray] at (-3/2,4/2) {};
\node (-24) [circle,draw,inner sep=3pt,line width=1,opacity=0.5,color=gray] at (-2/2,4/2) {};
\node (-24s) [circle,draw,inner sep=2pt,line width=1,opacity=0.5,color=gray] at (-2/2,4/2) {};
\node (-14) [circle,draw,inner sep=3pt,line width=1,opacity=0.5,color=gray] at (-1/2,4/2) {};
\node (-14s) [circle,draw,inner sep=2pt,line width=1,opacity=0.5,color=gray] at (-1/2,4/2) {};
\node (-33) [circle,draw,inner sep=3pt,line width=1,opacity=0.5,color=gray] at (-3/2,3/2) {};
\node (-33s) [circle,draw,inner sep=2pt,line width=1,opacity=0.5,color=gray] at (-3/2,3/2) {};
\node (-23) [circle,draw,inner sep=3pt,line width=1,opacity=1] at (-2/2,3/2) {};
\node (-23s) [circle,draw,inner sep=2pt,line width=1,opacity=1] at (-2/2,3/2) {};
\node (-13) [circle,draw,inner sep=3pt,line width=1,opacity=1] at (-1/2,3/2) {};
\node (-13s) [circle,draw,inner sep=2pt,line width=1,opacity=1] at (-1/2,3/2) {};
\node (-32) [circle,draw,inner sep=3pt,line width=1,opacity=0.5,color=gray] at (-3/2,2/2) {};
\node (-32s) [circle,draw,inner sep=2pt,line width=1,opacity=0.5,color=gray] at (-3/2,2/2) {};
\node (-22) [circle,draw,inner sep=3pt,line width=1,opacity=0.5,color=gray] at (-2/2,2/2) {};
\node (-22s) [circle,draw,inner sep=2pt,line width=1,opacity=0.5,color=gray] at (-2/2,2/2) {};
\node (-12) [circle,draw,inner sep=3pt,line width=1,opacity=0.5,color=gray] at (-1/2,2/2) {};
\node (-12s) [circle,draw,inner sep=2pt,line width=1,opacity=0.5,color=gray] at (-1/2,2/2) {};
\node (-31) [circle,draw,inner sep=3pt,line width=1,opacity=0.5,color=gray] at (-3/2,1/2) {};
\node (-31s) [circle,draw,inner sep=2pt,line width=1,opacity=0.5,color=gray] at (-3/2,1/2) {};
\node (-21) [circle,draw,inner sep=3pt,line width=1,opacity=0.5,color=gray] at (-2/2,1/2) {};
\node (-21s) [circle,draw,inner sep=2pt,line width=1,opacity=0.5,color=gray] at (-2/2,1/2) {};
\node (-11) [circle,draw,inner sep=3pt,line width=1,opacity=0.5,color=gray] at (-1/2,1/2) {};
\node (-11s) [circle,draw,inner sep=2pt,line width=1,opacity=0.5,color=gray] at (-1/2,1/2) {};

\node (15) [circle,draw,inner sep=3pt,line width=1,opacity=1] at (1/2,5/2) {};
\node (15s) [rectangle,draw,inner sep=2pt,line width=1,opacity=1] at (1/2,5/2) {};
\node (25) [circle,draw,inner sep=3pt,line width=1,opacity=1] at (2/2,5/2) {};
\node (25s) [rectangle,draw,inner sep=2pt,line width=1,opacity=1] at (2/2,5/2) {};
\node (14) [circle,draw,inner sep=3pt,line width=1,opacity=0.5,color=gray] at (1/2,4/2) {};
\node (14s) [rectangle,draw,inner sep=2pt,line width=1,opacity=0.5,color=gray] at (1/2,4/2) {};
\node (24) [circle,draw,inner sep=3pt,line width=1,opacity=0.5,color=gray] at (2/2,4/2) {};
\node (24s) [rectangle,draw,inner sep=2pt,line width=1,opacity=0.5,color=gray] at (2/2,4/2) {};
\node (13) [circle,draw,inner sep=3pt,line width=1,opacity=1] at (1/2,3/2) {};
\node (13s) [rectangle,draw,inner sep=2pt,line width=1,opacity=1] at (1/2,3/2) {};
\node (13b) [rectangle,draw,inner sep=8pt,line width=1,opacity=1,color=red] at (1/2,3/2) {};
\node (23) [circle,draw,inner sep=3pt,line width=1,opacity=1] at (2/2,3/2) {};
\node (23s) [rectangle,draw,inner sep=2pt,line width=1,opacity=1] at (2/2,3/2) {};
\node (12) [circle,draw,inner sep=3pt,line width=1,opacity=1] at (1/2,2/2) {};
\node (12s) [rectangle,draw,inner sep=2pt,line width=1,opacity=1] at (1/2,2/2) {};
\node (22) [circle,draw,inner sep=3pt,line width=1,opacity=0.5,color=gray] at (2/2,2/2) {};
\node (22s) [rectangle,draw,inner sep=2pt,line width=1,opacity=0.5,color=gray] at (2/2,2/2) {};
\node (11) [circle,draw,inner sep=3pt,line width=1,opacity=1] at (1/2,1/2) {};
\node (11s) [rectangle,draw,inner sep=2pt,line width=1,opacity=1] at (1/2,1/2) {};
\node (21) [circle,draw,inner sep=3pt,line width=1,opacity=0.5,color=gray] at (2/2,1/2) {};
\node (21s) [rectangle,draw,inner sep=2pt,line width=1,opacity=0.5,color=gray] at (2/2,1/2) {};

\node (-3-1) [rectangle,draw,inner sep=4pt,line width=1,opacity=0.5,color=gray] at (-3/2,-1/2) {};
\node (-3-1s) [circle,draw,inner sep=2pt,line width=1,opacity=0.5,color=gray] at (-3/2,-1/2) {};
\node (-2-1) [rectangle,draw,inner sep=4pt,line width=1,opacity=0.5,color=gray] at (-2/2,-1/2) {};
\node (-2-1s) [circle,draw,inner sep=2pt,line width=1,opacity=0.5,color=gray] at (-2/2,-1/2) {};
\node (-1-1) [rectangle,draw,inner sep=4pt,line width=1,opacity=0.5,color=gray] at (-1/2,-1/2) {};
\node (-1-1s) [circle,draw,inner sep=2pt,line width=1,opacity=0.5,color=gray] at (-1/2,-1/2) {};
\node (-3-2) [rectangle,draw,inner sep=4pt,line width=1,opacity=0.5,color=gray] at (-3/2,-2/2) {};
\node (-3-2s) [circle,draw,inner sep=2pt,line width=1,opacity=0.5,color=gray] at (-3/2,-2/2) {};
\node (-2-2) [rectangle,draw,inner sep=4pt,line width=1,opacity=0.5,color=gray] at (-2/2,-2/2) {};
\node (-2-2s) [circle,draw,inner sep=2pt,line width=1,opacity=0.5,color=gray] at (-2/2,-2/2) {};
\node (-1-2) [rectangle,draw,inner sep=4pt,line width=1,opacity=0.5,color=gray] at (-1/2,-2/2) {};
\node (-1-2s) [circle,draw,inner sep=2pt,line width=1,opacity=0.5,color=gray] at (-1/2,-2/2) {};
\node (-3-3) [rectangle,draw,inner sep=4pt,line width=1,opacity=0.5,color=gray] at (-3/2,-3/2) {};
\node (-3-3s) [circle,draw,inner sep=2pt,line width=1,opacity=0.5,color=gray] at (-3/2,-3/2) {};
\node (-2-3) [rectangle,draw,inner sep=4pt,line width=1,opacity=0.5,color=gray] at (-2/2,-3/2) {};
\node (-2-3s) [circle,draw,inner sep=2pt,line width=1,opacity=0.5,color=gray] at (-2/2,-3/2) {};
\node (-1-3) [rectangle,draw,inner sep=4pt,line width=1,opacity=0.5,color=gray] at (-1/2,-3/2) {};
\node (-1-3s) [circle,draw,inner sep=2pt,line width=1,opacity=0.5,color=gray] at (-1/2,-3/2) {};

\node (1-1) [rectangle,draw,inner sep=4pt,line width=1,opacity=1] at (1/2,-1/2) {};
\node (1-1s) [rectangle,draw,inner sep=2pt,line width=1,opacity=1] at (1/2,-1/2) {};
\node (2-1) [rectangle,draw,inner sep=4pt,line width=1,opacity=0.5,color=gray] at (2/2,-1/2) {};
\node (2-1s) [rectangle,draw,inner sep=2pt,line width=1,opacity=0.5,color=gray] at (2/2,-1/2) {};
\node (1-2) [rectangle,draw,inner sep=4pt,line width=1,opacity=1] at (1/2,-2/2) {};
\node (1-2s) [rectangle,draw,inner sep=2pt,line width=1,opacity=1] at (1/2,-2/2) {};
\node (2-2) [rectangle,draw,inner sep=4pt,line width=1,opacity=0.5,color=gray] at (2/2,-2/2) {};
\node (2-2s) [rectangle,draw,inner sep=2pt,line width=1,opacity=0.5,color=gray] at (2/2,-2/2) {};
\node (1-3) [rectangle,draw,inner sep=4pt,line width=1,opacity=0.5,color=gray] at (1/2,-3/2) {};
\node (1-3s) [rectangle,draw,inner sep=2pt,line width=1,opacity=0.5,color=gray] at (1/2,-3/2) {};
\node (2-3) [rectangle,draw,inner sep=4pt,line width=1,opacity=0.5,color=gray] at (2/2,-3/2) {};
\node (2-3s) [rectangle,draw,inner sep=2pt,line width=1,opacity=0.5,color=gray] at (2/2,-3/2) {};

\node (arrow) at (4.5/2,2/2) {\huge$\Rightarrow$};
\node (space) at (0/2,-4.5/2) {};
\end{tikzpicture}
}
{
\begin{tikzpicture}[baseline=(current bounding box.center)]
\node (nw) [inner sep=0] at (-4/2,6/2) {};
\node (n) [inner sep=0] at (0,6/2) {};
\node (ne) [inner sep=0] at (3/2,6/2) {};
\node (w) [inner sep=0] at (-4/2,0) {};
\node (c) [inner sep=0] at (0,0) {};
\node (e) [inner sep=0] at (3/2,0) {};
\node (sw) [inner sep=0] at (-4/2,-4/2) {};
\node (s) [inner sep=0] at (0,-4/2) {};
\node (se) [inner sep=0] at (3/2,-4/2) {};

\path (nw.center) edge[line width=1,opacity=1] (ne.center);
\path (w.center) edge[line width=1,dashed,opacity=1] (e.center);
\path (sw.center) edge[line width=1,opacity=1] (se.center);
\path (nw.center) edge[line width=1,opacity=1] (sw.center);
\path (n.center) edge[line width=1,dashed,opacity=1] (s.center);
\path (ne.center) edge[line width=1,opacity=1] (se.center);

\node (c1nw) at (-3.5/2,5.5/2) {};
\node (c1ne) at (-1.5/2,5.5/2) {};
\node (c1se) at (-1.5/2,4.5/2) {};
\node (c1sw) at (-3.5/2,4.5/2) {};
\path (c1nw.center) edge[line width=1,color=blue] (c1ne.center);
\path (c1ne.center) edge[line width=1,color=blue] (c1se.center);
\path (c1se.center) edge[line width=1,color=blue] (c1sw.center);
\path (c1sw.center) edge[line width=1,color=blue] (c1nw.center);

\node (c2nw) at (-2.5/2,3.5/2) {};
\node (c2ne) at (-0.5/2,3.5/2) {};
\node (c2se) at (-0.5/2,2.5/2) {};
\node (c2sw) at (-2.5/2,2.5/2) {};
\path (c2nw.center) edge[line width=1,color=blue] (c2ne.center);
\path (c2ne.center) edge[line width=1,color=blue] (c2se.center);
\path (c2se.center) edge[line width=1,color=blue] (c2sw.center);
\path (c2sw.center) edge[line width=1,color=blue] (c2nw.center);
\node (c2e) at (-0.5/2,3/2) {};
\path (c2e.center) edge[line width=2] (13b.west);
\node (c2s) at (-1.5/2,2.5/2) {};

\node (c3nw) at (0.5/2,-0.5/2) {};
\node (c3ne) at (1.5/2,-0.5/2) {};
\node (c3se) at (1.5/2,-2.5/2) {};
\node (c3sw) at (0.5/2,-2.5/2) {};
\path (c3nw.center) edge[line width=1,color=blue] (c3ne.center);
\path (c3ne.center) edge[line width=1,color=blue] (c3se.center);
\path (c3se.center) edge[line width=1,color=blue] (c3sw.center);
\path (c3sw.center) edge[line width=1,color=blue] (c3nw.center);
\node (c3n) at (1/2,-0.5/2) {};
\path (c3n.center) edge[line width=2] (13b.south);
\node (c3w) at (0.5/2,-1.5/2) {};

\node (c1label) at (-2.5/2,4.2/2) {isolated};
\node (c23label) at (-2/2,1/2) {dangling};
\path (c23label.north) edge[line width=1,->] (c2s.south);
\path (c23label.south) edge[line width=1,->] (c3w.west);
\node (cclabel) at (1.5/2,4.25/2) {$\begin{matrix}\text{connect--}\\\text{ing check}\end{matrix}$};

\node (-35) [circle,draw,inner sep=3pt,line width=1,opacity=1] at (-3/2,5/2) {};
\node (-35s) [circle,draw,inner sep=2pt,line width=1,opacity=1] at (-3/2,5/2) {};
\node (-25) [circle,draw,inner sep=3pt,line width=1,opacity=1] at (-2/2,5/2) {};
\node (-25s) [circle,draw,inner sep=2pt,line width=1,opacity=1] at (-2/2,5/2) {};
\node (-23) [circle,draw,inner sep=3pt,line width=1,opacity=1] at (-2/2,3/2) {};
\node (-23s) [circle,draw,inner sep=2pt,line width=1,opacity=1] at (-2/2,3/2) {};
\node (-13) [circle,draw,inner sep=3pt,line width=1,opacity=1] at (-1/2,3/2) {};
\node (-13s) [circle,draw,inner sep=2pt,line width=1,opacity=1] at (-1/2,3/2) {};

\node (13) [circle,draw,inner sep=3pt,line width=1,opacity=1] at (1/2,3/2) {};
\node (13s) [rectangle,draw,inner sep=2pt,line width=1,opacity=1] at (1/2,3/2) {};
\node (13b) [rectangle,draw,inner sep=7pt,line width=1,opacity=1,color=red] at (1/2,3/2) {};


\node (1-1) [rectangle,draw,inner sep=4pt,line width=1,opacity=1] at (1/2,-1/2) {};
\node (1-1s) [rectangle,draw,inner sep=2pt,line width=1,opacity=1] at (1/2,-1/2) {};
\node (1-2) [rectangle,draw,inner sep=4pt,line width=1,opacity=1] at (1/2,-2/2) {};
\node (1-2s) [rectangle,draw,inner sep=2pt,line width=1,opacity=1] at (1/2,-2/2) {};

\node (space) at (0/2,-4.5/2) {};
\end{tikzpicture}
}
\\

(a) Erasure subgraph and corresponding VH graph

{
\begin{tikzpicture}[baseline=(current bounding box.center)]
    \node (space1) at (0,2.5) {};
    
    \node (a1) [rectangle,draw,inner sep=6pt,line width=1] at (0,1) {};
    \node (a1s) [rectangle,draw,inner sep=4.5pt,line width=1] at (0,1) {};
    \node (a2) [rectangle,draw,inner sep=6pt,line width=1] at (0,0) {};
    \node (a2s) [rectangle,draw,inner sep=4.5pt,line width=1] at (0,0) {};
    
    \node (a3) [circle,draw,inner sep=5pt,line width=1] at (1,2) {};
    \node (a3s) [rectangle,draw,inner sep=4.5pt,line width=1] at (1,2) {};
    \node (a3box) [rectangle,draw,inner sep=9pt,line width=1,color=red] at (1,2) {};
    \node (a4) [circle,draw,inner sep=5pt,line width=1] at (1,1) {};
    \node (a4s) [rectangle,draw,inner sep=4.5pt,line width=1] at (1,1) {};
    \node (a5) [circle,draw,inner sep=5pt,line width=1] at (1,0) {};
    \node (a5s) [rectangle,draw,inner sep=4.5pt,line width=1] at (1,0) {};
    
    \path (a1) edge[line width=1] (a3);
    \path (a1) edge[line width=1] (a4);
    \path (a1) edge[line width=1] (a5);
    \path (a2) edge[line width=1] (a3);
    \path (a2) edge[line width=1] (a4);
    \path (a2) edge[line width=1] (a5);

    \node (pnw) at (-0.5,2.5) {};
    \node (pne) at (1.5,2.5) {};
    \node (pse) at (1.5,-0.5) {};
    \node (psw) at (-0.5,-0.5) {};
    \path (pnw.center) edge[line width=1,dotted,color=violet] (pne.center);
    \path (pne.center) edge[line width=1,dotted,color=violet] (pse.center);
    \path (pse.center) edge[line width=1,dotted,color=violet] (psw.center);
    \path (psw.center) edge[line width=1,dotted,color=violet] (pnw.center);

    \node (arrow) at (2.25,1) {\huge$\Rightarrow$};

    \node (b1) [rectangle,draw,inner sep=6pt,line width=1] at (3.25,1) {};
    \node (b1s) [rectangle,draw,inner sep=1.6pt,line width=1] at (3.25,1) {0};
    \node (b2) [rectangle,draw,inner sep=6pt,line width=1] at (3.25,0) {};
    \node (b2s) [rectangle,draw,inner sep=1.6pt,line width=1] at (3.25,0) {0};
    
    \node (b3) [circle,draw,inner sep=5pt,line width=1] at (4.25,2) {};
    \node (b3s) [rectangle,draw,inner sep=1.6pt,line width=1] at (4.25,2) {0};
    \node (b3box) [rectangle,draw,inner sep=9pt,line width=1,color=red] at (4.25,2) {};
    \node (b4) [circle,draw,inner sep=5pt,line width=1] at (4.25,1) {};
    \node (b4s) [rectangle,draw,inner sep=1.6pt,line width=1] at (4.25,1) {0};
    \node (b5) [circle,draw,inner sep=5pt,line width=1] at (4.25,0) {};
    \node (b5s) [rectangle,draw,inner sep=1.6pt,line width=1] at (4.25,0) {0};
    
    \path (b1) edge[line width=1] (b3);
    \path (b1) edge[line width=1] (b4);
    \path (b1) edge[line width=1] (b5);
    
    \path (b2) edge[line width=1] (b3);
    \path (b2) edge[line width=1] (b4);
    \path (b2) edge[line width=1] (b5);

    \node (or) at (5,1) {or};

    \node (c1) [rectangle,draw,inner sep=6pt,line width=1] at (5.75,1) {};
    \node (c1s) [rectangle,draw,inner sep=1.6pt,line width=1] at (5.75,1) {1};
    \node (c2) [rectangle,draw,inner sep=6pt,line width=1] at (5.75,0) {};
    \node (c2s) [rectangle,draw,inner sep=1.6pt,line width=1] at (5.75,0) {1};
    
    \node (c3) [circle,draw,inner sep=5pt,line width=1] at (6.75,2) {};
    \node (c3s) [rectangle,draw,inner sep=1.6pt,line width=1] at (6.75,2) {0};
    \node (b3box) [rectangle,draw,inner sep=9pt,line width=1,color=red] at (6.75,2) {};
    \node (c4) [circle,draw,inner sep=5pt,line width=1] at (6.75,1) {};
    \node (c4s) [rectangle,draw,inner sep=1.6pt,line width=1] at (6.75,1) {0};
    \node (c5) [circle,draw,inner sep=5pt,line width=1] at (6.75,0) {};
    \node (c5s) [rectangle,draw,inner sep=1.6pt,line width=1] at (6.75,0) {0};
    
    \path (c1) edge[line width=1] (c3);
    \path (c1) edge[line width=1] (c4);
    \path (c1) edge[line width=1] (c5);
    
    \path (c2) edge[line width=1] (c3);
    \path (c2) edge[line width=1] (c4);
    \path (c2) edge[line width=1] (c5);

    \node (space2) at (0,-0.5) {};
\end{tikzpicture}
}\\

(b) \textit{Frozen cluster} and possible solutions

{
\begin{tikzpicture}[baseline=(current bounding box.center)]
    \node (space1) at (0,2.5) {};
    
    \node (a1) [rectangle,draw,inner sep=6pt,line width=1] at (0,1) {};
    \node (a1s) [rectangle,draw,inner sep=4.5pt,line width=1] at (0,1) {};
    \node (a2) [rectangle,draw,inner sep=6pt,line width=1] at (0,0) {};
    \node (a2s) [rectangle,draw,inner sep=4.5pt,line width=1] at (0,0) {};
    
    \node (a3) [circle,draw,inner sep=5pt,line width=1] at (1,2) {};
    \node (a3s) [rectangle,draw,inner sep=4.5pt,line width=1] at (1,2) {};
    \node (a3box) [rectangle,draw,inner sep=9pt,line width=1,color=red] at (1,2) {};
    \node (a4) [circle,draw,inner sep=5pt,line width=1] at (1,1) {};
    \node (a4s) [rectangle,draw,inner sep=4.5pt,line width=1] at (1,1) {};
    \node (a5) [circle,draw,inner sep=5pt,line width=1] at (1,0) {};
    \node (a5s) [rectangle,draw,inner sep=4.5pt,line width=1] at (1,0) {};
    
    \path (a1) edge[line width=1] (a3);
    \path (a1) edge[line width=1] (a4);
    \path (a1) edge[line width=1] (a5);
    \path (a2) edge[line width=1] (a4);
    \path (a2) edge[line width=1] (a5);

    \node (pnw) at (-0.5,2.5) {};
    \node (pne) at (1.5,2.5) {};
    \node (pse) at (1.5,-0.5) {};
    \node (psw) at (-0.5,-0.5) {};
    \path (pnw.center) edge[line width=1,dotted,color=violet] (pne.center);
    \path (pne.center) edge[line width=1,dotted,color=violet] (pse.center);
    \path (pse.center) edge[line width=1,dotted,color=violet] (psw.center);
    \path (psw.center) edge[line width=1,dotted,color=violet] (pnw.center);

    \node (arrow) at (2.25,1) {\huge$\Rightarrow$};

    \node (b1) [rectangle,draw,inner sep=6pt,line width=1] at (3.25,1) {};
    \node (b1s) [rectangle,draw,inner sep=1.6pt,line width=1] at (3.25,1) {0};
    \node (b2) [rectangle,draw,inner sep=6pt,line width=1] at (3.25,0) {};
    \node (b2s) [rectangle,draw,inner sep=1.6pt,line width=1] at (3.25,0) {0};
    
    \node (b3) [circle,draw,inner sep=5pt,line width=1] at (4.25,2) {};
    \node (b3s) [rectangle,draw,inner sep=1.6pt,line width=1] at (4.25,2) {0};
    \node (b3box) [rectangle,draw,inner sep=9pt,line width=1,color=red] at (4.25,2) {};
    \node (b4) [circle,draw,inner sep=5pt,line width=1] at (4.25,1) {};
    \node (b4s) [rectangle,draw,inner sep=1.6pt,line width=1] at (4.25,1) {0};
    \node (b5) [circle,draw,inner sep=5pt,line width=1] at (4.25,0) {};
    \node (b5s) [rectangle,draw,inner sep=1.6pt,line width=1] at (4.25,0) {0};
    
    \path (b1) edge[line width=1] (b3);
    \path (b1) edge[line width=1] (b4);
    \path (b1) edge[line width=1] (b5);
    
    \path (b2) edge[line width=1] (b4);
    \path (b2) edge[line width=1] (b5);

    \node (or) at (5,1) {or};

    \node (c1) [rectangle,draw,inner sep=6pt,line width=1] at (5.75,1) {};
    \node (c1s) [rectangle,draw,inner sep=1.6pt,line width=1] at (5.75,1) {1};
    \node (c2) [rectangle,draw,inner sep=6pt,line width=1] at (5.75,0) {};
    \node (c2s) [rectangle,draw,inner sep=1.6pt,line width=1] at (5.75,0) {1};
    
    \node (c3) [circle,draw,inner sep=5pt,line width=1] at (6.75,2) {};
    \node (c3s) [rectangle,draw,inner sep=1.6pt,line width=1] at (6.75,2) {1};
    \node (b3box) [rectangle,draw,inner sep=9pt,line width=1,color=red] at (6.75,2) {};
    \node (c4) [circle,draw,inner sep=5pt,line width=1] at (6.75,1) {};
    \node (c4s) [rectangle,draw,inner sep=1.6pt,line width=1] at (6.75,1) {0};
    \node (c5) [circle,draw,inner sep=5pt,line width=1] at (6.75,0) {};
    \node (c5s) [rectangle,draw,inner sep=1.6pt,line width=1] at (6.75,0) {0};
    
    \path (c1) edge[line width=1] (c3);
    \path (c1) edge[line width=1] (c4);
    \path (c1) edge[line width=1] (c5);

    \path (c2) edge[line width=1] (c4);
    \path (c2) edge[line width=1] (c5);

    \node (space2) at (0,-0.5) {};
\end{tikzpicture}
}\\

(c) \textit{Free cluster} and possible solutions

\caption{{\bf (a)} Example showing how the VH graph is computed from the erasure subgraph (subgraph edges are excluded for simplicity). Nodes in the VH graph correspond to clusters of erased qubits (erased qubit-nodes in the same row or column sharing a check-node, indicated by a blue box above). Edges in the VH graph correspond to \textit{connecting checks} (check nodes adjacent to both a vertical and horizontal cluster, indicated by a red box above). \textit{Isolated clusters} have no connecting checks (degree 0 nodes in the VH graph). \textit{Dangling clusters} have exactly one connecting check (degree 1 nodes in the VH graph). A dangling cluster is determined to be \textit{frozen} or \textit{free} based on the connectivity of the subgraph induced by the cluster; (b) and (c) show two examples for the vertical dangling cluster in the purple box with different subgraphs. {\bf (b)} An example of a frozen cluster. All solutions for the cluster which satisfy the internal checks have the same contribution to the connecting check. {\bf (c)} An example of a free cluster. There exist solutions to the internal checks of the cluster which are 0 or 1 on the connecting check; equivalently, there exists an error on the cluster whose syndrome is non-zero only on the single connecting check.
}
\label{fig:VH_graph_example}
\end{figure*}

\begin{algo}{VH decoder}
\Input{An erasure vector $\varepsilon \in \Z_2^N$, a syndrome $s \in \Z_2^{R_Z}$.}
\Output{Either {\bf Failure} or an $X$-type error $\hat E \in \{I, X\}^N$ such that $\sigma(\hat E) = s$ and $\supp(\hat E) \subseteq \supp(\varepsilon)$.}
\label{algo:vh_decoder}

\BlankLine
	Set $\hat E = I$.\;
	Construct an empty stack $L = []$.\;
	\While{there exists an isolated or a dangling cluster $\kappa$}
	{
		\If{$\kappa$ is isolated or frozen}
		{
			Compute an error $\hat E_{\kappa}$ supported on $\kappa$ whose syndrome matches $s$ on the internal checks of $\kappa$ in $T(\HH_Z)$ (using the Gaussian decoder).\;
			Replace $\hat E$ by $\hat E \hat E_{\kappa}$ and $s$ by $s + \sigma(\hat E_{\kappa})$.\;
			{\bf For} all qubits $j$ in $\kappa$, set $\varepsilon_j = 0$.\;
		}
		\Else
		{
			Then $\kappa$ is free.\;
			Remove the free connecting check $c$ of $\kappa$ from the Tanner graph $T(\HH_Z)$.\;
			Add the pair $(\kappa, c)$ to the stack $L$.\;
			{\bf For} all qubits $j$ in $\kappa$, set $\varepsilon_j = 0$.\;
		}
	}
	\While{the stack $L$ is non-empty}
	{
		Pop a cluster $(\kappa, c)$ from the stack $L$.\;
		Add the check node $c$ to the Tanner graph $T(\HH_Z)$.\;
		Compute an error $\hat E_{\kappa}$ supported on $\kappa$ whose syndrome matches $s$ on all the checks of $\kappa$ in $T(\HH_Z)$, including the free check $c$ (using the Gaussian decoder).\;
		Replace $\hat E$ by $\hat E \hat E_{\kappa}$ and $s$ by $s + \sigma(\hat E_{\kappa})$.\;
	}
	{\bf if} $\varepsilon \neq 0$ {\bf return Failure}, {\bf else return $\hat E$.}
\end{algo}

A check node of $T(\HH_Z)$ that belongs to a single cluster is called an {\em internal check}, otherwise it is called a {\em connecting check}. From Proposition~\ref{prop:VH_graph_structure}, a connecting check must belong to one horizontal and one vertical cluster.

A cluster is said to be {\em isolated} if is has no connecting check.
Then, it can be corrected independently of the other clusters.
A {\em dangling cluster} is defined to be a cluster with a single connecting check.

Given a cluster $\kappa$, let $E(\kappa)$ be the set of errors supported on the qubits of $\kappa$ whose syndrome is trivial over the internal checks of $\kappa$.
Let $S(\kappa)$ be the set of syndromes of errors $E \in E(\kappa)$ restricted to the connecting checks of $\kappa$.
The set $E(\kappa)$ is a subset of $\{I, X\}^N$ and $S(\kappa)$ is a subset of $\{0, 1\}^{d(\kappa)}$ where $d(\kappa)$ is the number of connecting checks of the cluster $\kappa$.

A cluster $\kappa$ can have two types of connecting check. If $S(\kappa)$ contains a weight-one vector supported on a connecting check $c$, we say that $c$ is a {\em free check}. Otherwise, it is a {\em frozen check}.
If a check is free, the value of the syndrome on this check can be adjusted at the end of the procedure to match $s$ using an error included in the cluster $\kappa$.

To compute a correction $\hat E$ for a syndrome $s \in \Z_2^{R_Z}$, we proceed as follows.
Denote by $s_{\kappa}$ the restriction of $s$ to a cluster $\kappa$.
We initialize $\hat E = I$ and we consider three cases.

{\bf Case 1: Isolated cluster.}
If $\kappa$ is a isolated cluster, we use Gaussian elimination to find an error $\hat E_{\kappa}$ supported on the qubits of $\kappa$ whose syndrome matches $s$ on the internal checks of $\kappa$.
Then, we add $\hat E_{\kappa}$ to $\hat E$, we add $\sigma(\hat E_{\kappa})$ to $s$ and we remove $\kappa$ from the erasure $\varepsilon$.
This cluster can be corrected independently of the other cluster because it is not connected to any other cluster.

{\bf Case 2: Frozen dangling cluster.}
If $\kappa$ is a dangling cluster and its only connecting check is frozen, we proceed exactly as in the case of an isolated cluster. This is possible because any correction has the same contribution to the syndrome on the connecting check.

{\bf Case 3: Free dangling cluster.}
The correction of a dangling cluster $\kappa$ that contains a free check is delayed until the end of the procedure. We remove $\kappa$ from the erasure and we remove its free check from the Tanner graph $T(\HH_Z)$. Then, we look for a correction $\hat E'$ in the remaining erasure. We add $\hat E'$ to $\hat E$ and $\sigma(\hat E')$ to $s$. Once the remaining erasure is corrected and the syndrome is updated, we find a correction $\hat E_{\kappa}$ inside $\kappa$ that satisfies the remaining syndrome $s_{\kappa}$ in $\kappa$.
We proceed in that order because the value of the syndrome on a free check can be adjusted at the end of the procedure to match $s$ using an error included in the cluster $\kappa$ (by definition of free checks).

Altogether, we obtain the VH decoder (Algorithm~\ref{algo:vh_decoder}). Fig.~\ref{fig:VH_graph_example} shows a small example illustrating the core concepts introduced to explain the algorithm. Our implementation is available here~\cite{NC_github}. It works by correcting all isolated and dangling clusters until the erasure is fully corrected. Otherwise, it returns {\bf Failure}.


For a $r \times n$ matrix $H$, the complexity of the VH decoder is dominated by the cost of the Gaussian decoder which grows as $O(n^3)$ per cluster and there are at most $O(n)$ clusters with size linear in $n$.
Hence the overall cost is in $O(n^4)$ (assuming $r = O(n)$).
Therefore the VH decoder can be implemented in $O(N^2)$ bit operations where $N = \Theta(n^2)$ is the length of the quantum HGP code.
Using a probabilistic implementation of the Gaussian decoder~\cite{wiedemann1986solving,kaltofen1991wiedemann,lamacchia1990solving,kaltofen1995analysis}, we can implement the Gaussian decoder in $O(n^2)$ operations per cluster, reducing the complexity of the VH decoder to $O(n^3) = O(N^{1.5})$.

Algorithm~\ref{algo:vh_decoder} fails if the VH-graph of the erasure contains a cycle. However, one can modify the algorithm to eliminate some cycles by removing free checks of all clusters and not only dangling clusters. This may improve further the performance of the VH-decoder.

This modified version of the decoder would be applied after Algorithm 3 becomes stuck in a VH-decoder stopping set. Since the remaining clusters form a cycle, each of these has two or more connecting checks, and each connecting check can be identified as free or frozen with respect to a given cluster. Identifying and removing a free connecting check from the Tanner graph can possibly break the cycle, allowing Algorithm 3 to continue; as before, a correction matching the removed connecting check on the given cluster can be determined at the end of the algorithm. The additional cost of this modification comes from the need to classify multiple connecting checks per cluster, rather than a single check in the dangling case.

The identification of a connecting check as free or frozen itself requires an application of $O(n^3)$ complexity Gaussian elimination on a cluster of size $O(n)$, but the corresponding solutions can be saved to the stack until later used to find a partial correction supported on the cluster. This process could be applied as many as $r=O(n)$ times for a single cluster in the worst case scenario where all checks are connecting checks. Since the total number of clusters is at most $O(n)$, the extreme case where all clusters are contained in a cycle implies an upper bound on the complexity of $O(n^5)=O(N^{2.5})$, slower than the VH-decoder but still exceeding the Gaussian decoder. However, even the addition of this rule cannot account for all stopping sets. If the modified VH graph obtained after identifying and removing all free checks is still a cycle, then the decoder fails.

In comparison with our numerical results from Fig. ~\ref{fig:plots}, we see that the combination of pruned peeling and VH decoders performs almost as well as the ML decoder at low erasure erasure rates. This is to say that cycles of clusters, which are stopping sets for the VH decoder, are relatively infrequent in the low erasure rate regime. This behavior matches our intuition since errors for LDPC codes tend to be composed of disjoint small weight clusters~\cite{kovalev2013fault}.

\section{Conclusion}

We proposed a practical high-performance decoder for the correction of erasure with HGP codes.
Our numerical simulations show that the combination of the pruned peeling decoder with the VH decoder achieves a close-to-optimal performance. Moreover it can be implemented in complexity $O(N^2)$.
This decoder can be used as a subroutine of the BP-OSD decoder~\cite{panteleev2021degenerate}, the Union-Find decoder for surface or LDPC codes~\cite{huang2020fault,delfosse2021almost,delfosse2022toward}, or the Viderman's decoder~\cite{krishna2024viderman} to speed up these algorithms.

Combination with the Union-Find decoder also suggests a natural method by which our techniques could be generalized to a mixed error channel, with both erasures and bit/phase flips. Erased qubits can still be replaced with uniformly random mixed states, but we relax the initial assumption that non-erased bits do not have errors. After making a syndrome measurement, we may apply the Union-Find algorithm to grow clusters of qubits around the unsatisfied checks until clusters are large enough to support a correction. Finally, we treat these clusters as an erasure pattern and assume that qubits not contained in this erasure do not have errors. In this way, the mixed error problem can be converted into a simpler erasure error problem, allowing the application of the Pruned Peeling + VH decoder in the case of HGP codes. The Union-Find decoder for surface codes~\cite{huang2020fault,delfosse2021almost} already has almost linear complexity, but the combination with Pruned Peeling + VH could offer an improvement for the Union-Find algorithm generalized to quantum LDPC codes~\cite{delfosse2022toward}.

In future work, it would be interesting to adapt our decoder to other quantum LDPC codes~\cite{breuckmann2021balanced, panteleev2022asymptotically, leverrier2022quantum, dinur2023good}.
We are also wondering if one can reduce the complexity further to obtain a linear time ML decoder for the correction of erasure.
Finally, it would be interesting to investigate the resource overhead of various quantum computing architectures capable of detecting erasures~\cite{wu2022erasure,kang2023quantum,kubica2023erasure,tsunoda2023error}.

\begin{acknowledgements}
This research was supported by the MSR-Inria Joint Centre.
AL acknowledges support from the Plan France 2030 through the project ANR-22-PETQ-0006.
\end{acknowledgements}

\bibliographystyle{quantum}
\bibliography{references.bib}

\end{document}